\documentclass[11pt]{article}

\usepackage[margin=1in]{geometry}
\usepackage{amsmath,amssymb,amsthm,mathtools}
\usepackage{microtype}
\usepackage{enumitem}
\usepackage{booktabs}
\usepackage{natbib}
\usepackage{xcolor}
\usepackage{hyperref}
\usepackage{cleveref}
\usepackage[ruled,vlined,linesnumbered]{algorithm2e}

\hypersetup{
  colorlinks=true,
  linkcolor=blue!60!black,
  citecolor=blue!60!black,
  urlcolor=blue!60!black
}

\newtheorem{theorem}{Theorem}[section]
\newtheorem{lemma}{Lemma}[section]
\newtheorem{proposition}{Proposition}[section]

\theoremstyle{definition}

\newcommand{\plu}{\operatorname{plu}}
\newcommand{\IV}{\operatorname{IV}}
\newcommand{\ML}{\operatorname{ML}}
\newcommand{\SC}{\operatorname{SC}}
\newcommand{\dist}{\operatorname{dist}}
\newcommand{\E}{\mathbb{E}}
\newcommand{\Prb}{\mathbb{P}}
\newcommand{\R}{\mathbb{R}}
\newcommand{\bottom}{\operatorname{bottom}}
\newcommand{\MIV}{\operatorname{MIV}}
\Crefname{algocf}{Algorithm}{Algorithms}
\Crefname{lemma}{Lemma}{Lemmas}

\title{An improved bound for the randomized metric distortion problem}
\author{Fabian Frank\footnote{E-Mail: \url{fabian_w.frank@tum.de}}\\
Technische Universität München}
\date{\today}

\begin{document}
\maketitle

\begin{abstract}
We propose a randomized social choice rule called Mixed Integrated Veto (MIV)
with metric distortion of $5/2$, improving the previous best upper bound of
$2.75271$.  MIV is the equal mixture of Maximal Lotteries and Integrated Veto, a
new rule built on the Simultaneous Veto process of Kizilkaya and Kempe.  Rather
than returning the candidate surviving longest, Integrated Veto assigns each
candidate probability proportional to its average score over the whole process.
\end{abstract}

\section{Introduction}
In the metric distortion model it is assumed that voters and candidates lie in a
common unknown metric space, and a voter's cost for a candidate is equal to
their distance. A voting rule observes only each voter's ranking of the
candidates by distance and outputs a lottery.  
Given a profile and a metric consistent with it, the distortion of a voting rule is the ratio between the expected social cost of the lottery it outputs and the cost of the best candidate in that metric. The metric distortion of the rule is the worst case over all profiles and all consistent metrics.
 
For deterministic rules the optimal metric distortion is $3$ \citep{ABEPS18,GHS20},
while for randomized rules it has remained open.
The best known lower bound is
$2.1126$ \citep{CR22}, and the best known upper bound is due to \citet{CRWW24}, who obtained $2.75271$ by randomizing between Maximal Lotteries --- the
equilibria of the zero-sum game induced by pairwise majority margins --- and
Random Dictatorship on the weighted uncovered set. 

We improve this bound to $5/2$ by using a different mixture. Instead of Random Dictatorship, we mix a maximal lottery with a new rule, \emph{Integrated Veto}, which is built on the Simultaneous Veto process of \citet{KK23}. In that process, each candidate with positive plurality score starts
with that score, and every voter continuously decrements the score of her least
preferred candidate among those still active, until no score remains. Rather
than returning the candidate surviving longest, Integrated Veto assigns each
candidate probability proportional to the average score it held over the process.
 
Our analysis rests on a certificate of \citet{CRWW24}, which bounds the metric distortion of a lottery by comparing the probability it places on candidates that a given set is collectively preferred to against how far that set is from containing a candidate ranked above everything outside it by every voter. We show that Integrated Veto satisfies the certificate well when the latter quantity is small, whereas Maximal Lotteries do so when it is large. Mixing Integrated Veto and Maximal Lotteries in equal parts guarantees a metric distortion of $5/2$, and we show that this bound is tight for the mixture.

\section{Related work}
The metric distortion framework was introduced by \citet{ABP15,ABEPS18}, who showed that no deterministic rule can guarantee
distortion below $3$ and that the uncovered set achieves $5$. 
\citet{MW19} improved the upper bound to $2+\sqrt{5}$ before \citet{GHS20} closed the gap with Plurality Matching, a
deterministic rule achieving the optimal bound of $3$.
\citet{KK22} subsequently introduced Plurality Veto, a 
simpler rule with the same guarantee. 

In the randomized setting, \citet{AP16,AP17} showed that Random Dictatorship achieves a metric distortion of $3-\tfrac2n$, which converges to $3$ as $n$ grows, and that no randomized rule achieves distortion below $2$, even on the
line.
\citet{CR22} improved this lower bound to $2.1126$. \citet{CRWW24} were the first to break the barrier of $3$ by proposing a rule that guarantees $2.75271$. Their analysis proceeds
through a biased-metric certificate which we use as
Proposition~\ref{prop:crww}. Recently, \citet{CGRW26} showed that uniformly sampling from a deterministically chosen constant-size multiset of candidates already achieves a metric distortion of $3- \varepsilon$ for some constant $\varepsilon > 0$.

\paragraph{Veto-based rules and the veto core.}
Plurality Veto initializes each candidate with its plurality score and decrements it according to voters' bottom choices among the candidates still
standing \citep{KK22}. \citet{KK23} later introduced the
generalized veto core, which extends Moulin's proportional veto core \citep{Mou81}, together with a continuous Simultaneous Veto formulation
in which all voters decrement simultaneously and candidates are eliminated as
their support is exhausted.
The construction has since been extended, for instance to the
$(p,q)$-veto core in the learning-augmented setting \citep{BFGT24} and to $k$-approval variants trading distortion against minority
protection \citep{KK25}.  Integrated Veto inherits the score dynamics of Simultaneous Veto, but differs in its output: rather than returning the final survivor, it also takes the intermediate states of the process into account.

\paragraph{Maximal Lotteries.}
Maximal Lotteries go
back to \citet{Kre65} and \citet{Fis84} and have been widely studied
in probabilistic social choice \citep{Bra17}.  \citet{CRWW24} were the first to consider them in the metric
distortion setting.

\section{Preliminaries}
Let $N$ be a finite set of voters and $C$ a finite set of candidates.  Every
voter $v\in N$ has a strict linear order $\succ_v$ over $C$, and we write
$P=(\succ_v)_{v\in N}$ for the resulting \emph{preference profile},
$\Delta(C)$ for the set of lotteries over $C$, and $I^c:=C\setminus I$ for $I\subseteq C$.  Throughout,
$\Prb_v[\,\cdot\,]$ and $\E_v[\,\cdot\,]$ denote probability and expectation
with respect to a uniformly random voter $v\in N$; we use this notation because our arguments only ever refer to the fraction of voters with a given preference, never to their number.

A \emph{pseudo-metric} on $N\cup C$ is a map
$d\colon(N\cup C)\times(N\cup C)\to\R_{\ge0}$ with $d(x,x)=0$, $d(x,y)=d(y,x)$,
and $d(x,z)\le d(x,y)+d(y,z)$ for all $x,y,z\in N\cup C$. It is \emph{consistent} with $P$, written $d\models P$, if $a\succ_v b$
implies $d(v,a)\le d(v,b)$ for all $v\in N$ and $a,b\in C$.
The \emph{social cost} of a candidate $c\in C$ under $d$ is defined as
\[
  \SC(c,d):=\E_v\bigl[d(v,c)\bigr]. 
\]
If the metric is clear from context we omit $d$ and simply write $\SC(c)$.

An ordinal randomized social choice function $f$ maps every profile $P$ to a
lottery $f(P)\in\Delta(C)$.  The \emph{metric distortion of a lottery
$p\in\Delta(C)$ on a profile $P$} is
\[
  \dist(p,P):=\sup_{d\,\models\,P}\
  \frac{\E_{c\sim p}\bigl[\SC(c,d)\bigr]}{\min_{c\in C}\SC(c,d)} ,
\]
and the \emph{metric distortion} of $f$ is
\[
  \dist(f):=\sup_{C,\,P}\ \dist\bigl(f(P),P\bigr),
\]
where the supremum ranges over all finite candidate sets $C$ and all profiles
$P$ over $C$.

The remaining notation follows \citet{CRWW24}.
Throughout the remaining part of the paper, 
$I$ denotes an arbitrary set with $ \emptyset \subsetneq I \subsetneq C$.
For candidates $i,j\in C$ we write
\[
  s_{i\succ j}:=\Prb_v[i\succ_v j\,],
  \qquad
  s_{I\succ j}:=\Prb_v[i\succ_v j\ \ \forall i\in I\,],
  \qquad
  s_{i\succ I^c}:=\Prb_v[i\succ_v j\ \ \forall j\in I^c].
\]

We extend this notation to lotteries. For $p,p'\in\Delta(C)$ let
\[
  s_{p\succ p'}:=\Prb_{a\sim p,\,b\sim p',\,v}\bigl[a\succ_v b\bigr],
\]
where, following \citet{CRWW24}, a tie $a=b$ counts as $\tfrac12$, so that
$s_{p\succ p'}+s_{p'\succ p}=1$.  We identify a candidate $c$ with the lottery
$\delta_c$ that puts probability $1$ on $c$, so that expressions such as
$s_{p\succ c}$ are defined; note that under this identification
$s_{c\succ c}=\tfrac12$, whereas $s_{c\succ c}=0$ under the definition for
candidates above.  The convention is only used for comparisons involving
lotteries.

Further, let $\plu(j):=\Prb_v[j\text{ is }v\text{'s top choice}\,]$ with
$\plu(S):=\sum_{j\in S}\plu(j)$ for $S\subseteq C$.  For $x\in\R^C_{\geq 0}$ set
\begin{equation}
  \ell_I(x):=\sum_{j\in I^c}s_{I\succ j}x(j),
  \qquad
  q(I):=\min_{i\in I}\bigl(1-s_{i\succ I^c}\bigr)\in[0,1].
  \label{eq:ellq}
\end{equation}
As for $\plu$, we write $p(S):=\sum_{j\in S}p(j)$ for $S\subseteq C$.

For a lottery $p$, the quantity $\ell_I(p)$ is the probability $p$ places outside $I$, each candidate
weighted by the fraction of voters preferring all of $I$ to it, and $q(I)$
measures how close the strongest member of $I$ comes to being ranked above all
of $I^c$ by every voter. We define $\ell_I$ on all of $\R^C_{\ge0}$ rather than only on $\Delta(C)$ because we apply it to the probability that Integrated Veto adds in a single phase, which need not sum to one.

Note that $\ell_I$ is linear: for $x,y\in\R^C_{\ge0}$ and
$\lambda,\lambda'\in\R_{\ge0}$,
\[
  \ell_I(\lambda x+\lambda'y)
  =\sum_{j\in I^c}s_{I\succ j}\bigl(\lambda x(j)+\lambda'y(j)\bigr)
  =\lambda\sum_{j\in I^c}s_{I\succ j}x(j)+\lambda'\sum_{j\in I^c}s_{I\succ j}y(j)
  =\lambda\,\ell_I(x)+\lambda'\,\ell_I(y),
\]
and the same holds for finite sums.

Finally, a \emph{maximal lottery} is a lottery $p \in\Delta(C)$ with
\[
  \sum_{i\in C}p(i)\bigl(s_{i\succ j}-s_{j\succ i}\bigr)\ \ge\ 0
  \qquad\text{for every } j\in C,
\]
that is, an equilibrium strategy of the symmetric zero-sum game with payoff
matrix $\bigl(s_{i\succ j}-s_{j\succ i}\bigr)_{i,j\in C}$.  As this matrix is
skew-symmetric, a maximal lottery exists for every profile \citep{Fis84}.  If
some $b\in C$ is a strict Condorcet winner, i.e.\ $s_{b\succ j}>\tfrac12$ for
all $j\neq b$, then the unique maximal lottery puts probability $1$ on $b$.  We write 
$\ML(P)$ for an arbitrary maximal lottery of 
$P$, fixing one for each profile; all our results hold regardless of the choice.

In order to prove our main statement we use the following result by \citet{CRWW24}.

\begin{proposition}[{\citet[Sec. 3]{CRWW24}}]
\label{prop:crww}
Let $P$ be a profile, $p\in\Delta(C)$ and $\lambda\ge0$.  If
$\ell_I(p)\le\lambda\bigl(1-s_{i^*\succ I^c}\bigr)$ for every non-empty proper
$I\subsetneq C$ and every $i^*\in I$, then $\dist(p,P)\le 1+2\lambda$. 
\end{proposition}

Since $q(I)=\min_{i^*\in I}\bigl(1-s_{i^*\succ I^c}\bigr)$, the hypothesis of
\Cref{prop:crww} is equivalent to $\ell_I(p)\le\lambda q(I)$ for every $I$.
We use this as follows: for Integrated Veto and for Maximal Lotteries we bound
$\ell_I$ in terms of $q(I)$ separately, and then combine the two bounds into a
single $\lambda=\tfrac{3}{4}$ that works for every $I$.

\section{Integrated Veto}

\begin{algorithm}[t]
\DontPrintSemicolon
\SetKwInOut{Input}{Input}\SetKwInOut{Output}{Output}
\Input{A profile $P$ over candidates $C$.}
\Output{A lottery $p\in\Delta(C)$.}
$k\leftarrow0$\;
$r_j^0\leftarrow\plu(j)$ and $p(j)\leftarrow0$ for all $j\in C$\;
$A_0\leftarrow\{j\in C: r_j^0>0\}$\;
\While{$A_k\neq\varnothing$}{
  $b_j^k\leftarrow\Prb_v\!\left[j=\bottom_v(A_k)\right]$ for all $j\in A_k$\;
  $\Delta_k\leftarrow\min\bigl\{\,r_j^k/b_j^k \;:\; j\in A_k,\ b_j^k>0\,\bigr\}$\;
  \ForEach{$j\in A_k$}{
    $r_j^{k+1}\leftarrow r_j^k-\Delta_k b_j^k$\;
    $p(j)\leftarrow p(j)+\Delta_k\bigl(r_j^k+r_j^{k+1}\bigr)$\;
  }
  $A_{k+1}\leftarrow\{j\in A_k: r_j^{k+1}>0\}$\;
  $k\leftarrow k+1$\;
}
$K\leftarrow k$\;
\Return $p$\;
\caption{Integrated Veto}
\label{alg:iv}
\end{algorithm}

We now analyze the Integrated Veto rule (IV), which is formally defined in
\Cref{alg:iv} and very similar to the Simultaneous Veto rule of \citet{KK23}.

For a non-empty set $A\subseteq C$ let $\bottom_v(A)$ denote voter $v$'s
least-preferred candidate in $A$, which is well defined since $\succ_v$ is a
strict linear order.  We refer to one execution of the \textbf{while} loop of
\Cref{alg:iv} as a \emph{phase} and write $K$ for the total number of phases;
for $0\le k\le K$, $r_j^k$, $b_j^k$, $\Delta_k$ and $A_k$ denote the
values these variables hold at the start of phase $k$.  Note that $r_j^k=0$ for
every $j\notin A_k$, and that $\Delta_k>0$, since $r_j^k>0$ for all $j\in A_k$.

Each candidate $j$ starts with a score $r_j^0=\plu(j)$, and the active set $A_k$
contains the candidates whose score is still positive; in particular, a
candidate who is never ranked first never becomes active and therefore receives
probability $0$.  In phase $k$, every voter vetoes her least preferred candidate
among the active candidates $A_k$, so candidate $j$ loses score at rate $b_j^k$,
the fraction of voters ranking $j$ last among $A_k$.  The phase lasts until the
first active candidate runs out of score, which happens after time $\Delta_k$;
every active candidate's score then drops by $\Delta_kb_j^k$, and every
candidate whose score has reached $0$ leaves the active set.

During the phase the score of $j$ falls linearly from $r_j^k$ to $r_j^{k+1}$, so
its average over the phase is $\tfrac12(r_j^k+r_j^{k+1})$ and the amount added
to $p(j)$ is twice this average, multiplied by the length $\Delta_k$ of the
phase.  Next, we show that $p$ is indeed a lottery.

\begin{lemma}
    \Cref{alg:iv} returns a lottery.
\end{lemma}
\begin{proof}
Let $p=\IV(P)$ be the outcome of \Cref{alg:iv} on a profile $P$.  We first show
$r_j^k\ge0$ for all $j$ and $k$.  If $j\notin A_k$ or $b_j^k=0$, nothing is
subtracted in phase $k$.  If $j\in A_k$ and $b_j^k>0$, the choice of $\Delta_k$
gives $\Delta_k\le r_j^k/b_j^k$ and hence
$r_j^{k+1}=r_j^k-\Delta_kb_j^k\ge0$.  Since $r_j^0=\plu(j)\ge0$, induction on
$k$ gives the claim.  As $p(j)$ is only ever increased, by the non-negative
amount $\Delta_k(r_j^k+r_j^{k+1})$, we get $p(j)\ge0$ for every $j$.
 
Next we show that $p$ sums to $1$.  Set $S_k:=\sum_jr_j^k$.  First, observe that in each iteration the algorithm removes all candidates with a minimal $\frac{r_j^k}{b_j^k}$ ratio from the active set, thus the number of iterations $K$ is bounded by $\lvert A_0 \rvert \leq \lvert C \rvert$.  Moreover
$\sum_jb_j^k=1$, since every voter has exactly one least preferred candidate in
$A_k$ as long as $A_k\neq\varnothing$, and hence
$S_{k+1}=\sum_j\bigl(r_j^k-\Delta_kb_j^k\bigr)=S_k-\Delta_k\sum_jb_j^k=S_k-\Delta_k$.

Thus, 
\[
  \sum_j p(j)=\sum_k\Delta_k\bigl(S_k+S_{k+1}\bigr)
  =\sum_k\bigl(S_k-S_{k+1}\bigr)\bigl(S_k+S_{k+1}\bigr)
  =\sum_k\bigl(S_k^2-S_{k+1}^2\bigr)=S_0^2-S_K^2=1 ,
\]
using $S_0=\sum_j\plu(j)=1$ and $S_K=0$, since the active set is empty after $K$ phases.
\end{proof}

Next, in order to bound the metric distortion of IV we want to use \Cref{prop:crww}. For this we use the following lemma.

\begin{lemma}
\label{lem:IVbound}
Let $P$ be a profile and $p=\IV(P)$.  Then $\ell_I(p)\le\plu(I^c)^2\le q(I)^2$ 
 for every non-empty proper $I\subsetneq C$.
\end{lemma}

\begin{proof}
Fix $I$ and set $R_k:=\sum_{j\in I^c}r_j^k$, $B_k:=\sum_{j\in I^c}b_j^k$.
Then, since $r_j^{k+1}= r_j^{k}-\Delta_kb_j^k $ we get that
$R_k-R_{k+1}=\Delta_kB_k$.
Next, we claim that $
  s_{I\succ j}\le B_k
  \text{ for all }j\in A_k\cap I^c .
$
For this observe that
if $j\in A_k\cap I^c$ and $v$ is counted by $s_{I\succ j}$, then every active
member of $I$ is ranked by $v$ above the active candidate $j$, so
$\bottom_v(A_k)\notin I$. 
Hence $\bottom_v(A_k)\in A_k\cap I^c$, so every voter counted by $s_{I\succ j}$
contributes to $\sum_{j'\in I^c}b_{j'}^k=B_k$.

Candidates outside $A_k$ satisfy $r_j^k=r_j^{k+1}=0$ and hence contribute
nothing to this sum, and the factors $r_j^k+r_j^{k+1}$ are non-negative. Therefore the contribution of phase
$k$ to $\ell_I(p)$ is at most
\[
  \sum_{j\in I^c}s_{I\succ j}\,\Delta_k\bigl(r_j^k+r_j^{k+1}\bigr)
  \ \le\ \Delta_kB_k\bigl(R_k+R_{k+1}\bigr)
  =\bigl(R_k-R_{k+1}\bigr)\bigl(R_k+R_{k+1}\bigr)
  =R_k^2-R_{k+1}^2 .
\]

Since $p(j)=\sum_k\Delta_k\bigl(r_j^k+r_j^{k+1}\bigr)$ by \Cref{alg:iv} and
$\ell_I$ is linear, summing over all phases gives
\[
  \ell_I(p)\ \le\ \sum_k\bigl(R_k^2-R_{k+1}^2\bigr)
  \ =\ R_0^2-R_K^2\ =\ \plu(I^c)^2 ,
\]
using $R_K=0$ and $R_0=\plu(I^c)$.

Finally, we claim that 
$\plu(I^c)\le q(I)$. 
Fix $i \in I$.
A voter whose top choice lies in $I^c$ ranks some candidate of
$I^c$ above $i$, so $\plu(I^c)\le1-s_{i\succ I^c}$.
Minimizing over $i\in I$ proves the claim.
\end{proof}

\section{Mixed Integrated Veto}
 
Next, we define the \emph{Mixed Integrated Veto rule} (MIV) as
$\MIV(P)=\tfrac12\bigl(\IV(P)+\ML(P)\bigr)$.
 
For a lottery $p\in\Delta(C)$ and $X\subseteq C$ with $p(X)>0$, let $p_{\mid X}$
denote $p$ conditioned on $X$, that is
\[
  p_{\mid X}(c):=
  \begin{cases}
    p(c)/p(X), & c\in X,\\
    0, & c\notin X .
  \end{cases}
\]
In particular $p=p(X)\,p_{\mid X}+p(X^c)\,p_{\mid X^c}$ whenever both parts are
defined.
 
In order to analyze the metric distortion of MIV we use the following result from
\citet{CRWW24}.
 
\begin{lemma}[Theorem~1 of \citealp{CRWW24}\footnote{\citet{CRWW24} state this for the level sets $I_t=\{j:x_j\le t\}$ of a biased
metric, but every $I$ arises as such a level set: choosing any $i^*\in I$ as the
distinguished candidate and setting $x_i=0$ for $i\in I$ and $x_j=1$ for
$j\notin I$ gives $I_t=I$ for all $t\in(0,1)$.}]
\label{lem:MLboundCRWW}
Let $p$ be a maximal lottery. Then for every non-empty proper $I\subsetneq C$ it holds that $\ell_I(p)\le\tfrac12\,p(I^c)$.
\end{lemma}
 
We further record two properties of $s$, both of which appear as Claims 1 and~2
of \citet{CRWW24}.
 
\begin{lemma}[Claim 1 of \citealp{CRWW24}]
\label{lem:spropsTwo}
If $p$ is a maximal lottery with $p(I),p(I^c)>0$, then
$s_{p_{\mid I}\succ p_{\mid I^c}}\le\tfrac12$.
\end{lemma}

\begin{lemma}[Claim 2 of \citealp{CRWW24}]
\label{lem:sprops}
For all $p,p',p''\in\Delta(C)$ we have $s_{p\succ p''}\le s_{p\succ p'}+s_{p'\succ p''}$.
\end{lemma}

Combining these, we get the following bound.
 
\begin{lemma}
\label{lem:MLbound}
Let $p$ be a maximal lottery.  Then for every non-empty proper $I\subsetneq C$,
\[
  \ell_I(p)\le\min\bigl\{q(I),\ \tfrac12\,p(I^c)\bigr\}
  \ \le\ \min\bigl\{q(I),\tfrac12\bigr\}.
\]
\end{lemma}
 
\begin{proof}
Put $q_i:=1-s_{i\succ I^c}$ for $i\in I$, so that $q(I)=\min_{i\in I}q_i$.  By
\Cref{lem:MLboundCRWW} it suffices to show $\tfrac12 p(I^c)\le q_{i^*}$ for an
arbitrary $i^*\in I$, which we now fix.

If $p(I^c)=0$ this is immediate.  If $p(I)=0$, then $p(I^c)=1$, so
$s_{i^*\succ I^c}\le\min_{j\in I^c}s_{i^*\succ j}\le s_{i^*\succ p}\le\tfrac12$,
where the last step uses that $p$ is a maximal lottery; hence
$q_{i^*}\ge\tfrac12=\tfrac12\,p(I^c)$.

Assume now $p(I),p(I^c)>0$.
From $s_{i^*\succ I^c}\le\min_{j\in I^c}s_{i^*\succ j}\le s_{i^*\succ p_{\mid I^c}}$ we get
$s_{p_{\mid I^c}\succ i^*}\le q_{i^*}$, and \Cref{lem:spropsTwo} and \Cref{lem:sprops} give
$s_{p_{\mid I}\succ i^*}\le s_{p_{\mid I} \succ p_{\mid I^c}}+s_{p_{\mid I^c}\succ i^*}\le\tfrac12+q_{i^*}$.
Decomposing $p$ over $I$ and $I^c$ and using that $p$ is a maximal lottery,
\[
  \tfrac12\le s_{p\succ i^*}=p(I^c)\,s_{p_{\mid I^c}\succ i^*}+p(I)\,s_{p_{\mid I} \succ i^*}
  \le p(I^c)\,q_{i^*}+p(I)\bigl(\tfrac12+q_{i^*}\bigr)
  =\tfrac12+q_{i^*}-\tfrac{p(I^c)}{2},
\]
so $\tfrac12\,p(I^c)\le q_{i^*}$.  The second inequality of the statement
follows from $p(I^c)\le1$.
\end{proof}
 
We can now combine \Cref{lem:IVbound} and \Cref{lem:MLbound} to bound the metric
distortion of MIV.

\begin{theorem}
\label{thm:main}
MIV has a metric distortion of at most $5/2$.
\end{theorem}

\begin{proof}
Consider any profile $P$, set $p=\IV(P)$ and $p'=\ML(P)$, and fix $I$.  By
linearity of $\ell_I$ and \Cref{lem:IVbound} and \Cref{lem:MLbound},
\[
  \ell_I\bigl(\tfrac12 p+\tfrac12 p'\bigr)
  \le\tfrac12 q(I)^2+\tfrac12\min\bigl\{q(I),\tfrac12\bigr\}.
\]
We distinguish two cases.  If $0\le q(I)\le\tfrac12$, the right-hand side is at
most $\tfrac12 q(I)^2+\tfrac12 q(I)\le\tfrac34 q(I)$; if
$\tfrac{1}{2}\le q(I)\le1$, it is at most
$\tfrac{1}{2} q(I)^2+\tfrac{1}{4}\le\tfrac{3}{4} q(I)$.

Since $I$ was arbitrary, \Cref{prop:crww} with $\lambda=\tfrac34$ gives
$\dist(\MIV(P),P)\le 1+2\cdot\tfrac{3}{4}=5/2$, and since $P$ was arbitrary,
$\dist(\MIV)\le 5/2$.
\end{proof}

Finally, we also show that the bound is tight by providing a family of profiles with consistent metrics for which the social cost ratio between the lottery returned by $\MIV$ and an optimal candidate converges to $5/2$.

\begin{theorem}
\label{thm:other}
$\MIV$ has a metric distortion of at least $5/2$.
\end{theorem}

\begin{proof}
Fix $\varepsilon\in(0,\tfrac12)$ and $r\ge3$. Then, we define the following profile $P$ with $C=\{a,b,c_1,\dots,c_r\}$.
For each $k\in[r]$, a mass $(1/2+\varepsilon)/r$ of voters ranks
$c_k\succ b\succ a\succ C\setminus\{a,b,c_k\}$ and a mass $(1/2-\varepsilon)/r$
ranks $c_k\succ a\succ b\succ C\setminus\{a,b,c_k\}$, where the ordering of 
$C\setminus\{a,b,c_k\}$ is arbitrary.

Define $p = \IV(P)$.
Only the $c_j$ have positive plurality score, so $A_0=\{c_1,\dots,c_r\}$,
$p(a)=p(b)=0$ and $\sum_{j}p(c_j)=1$.  Since $s_{b\succ c_j}=1-1/r>\tfrac12$
and $s_{b\succ a}=\tfrac12+\varepsilon>\tfrac12$, candidate $b$ is a strict
Condorcet winner, so $\ML(P)=\delta_b$ and
$\MIV(P)=\frac{1}{2} \delta_b+\frac{1}{2} p$. 

Consider the metric $d$ with $d(v,a)=1$ for every voter $v$, and, for $c\neq a$,
$d(v,c)=1$ if $c\succ_v a$ and $d(v,c)=3$ otherwise; all remaining distances are
given by the shortest-path closure. 
Observe that $d$ is consistent with the profile and the closure leaves the distances above
unchanged: any path from a voter to a candidate other than the direct edge
alternates between voters and candidates, hence uses at least three edges, each
of length at least $1$.
A voter ranks $c_j$ above $a$ exactly
when $c_j$ is her top choice, and ranks $b$ above $a$ exactly when she is one of
the $(1/2+\varepsilon)$ voters, so
\[
  \SC(a)=1,\qquad \SC(c_j)=3-\tfrac2r \qquad(\text{for } j \in [r]),\qquad \SC(b)=2-2\varepsilon .
\]
As $\varepsilon<\tfrac12$ and $r\ge3$, candidate $a$ minimizes social cost, and
\[
  \dist(\MIV(P),P) \ge \frac{\tfrac12\SC(b)+\tfrac{1}{2}\sum_j p(c_j)\SC(c_j)}{\SC(a)}
  \ = \tfrac12(2-2\varepsilon)+\tfrac{1}{2}\bigl(3-\tfrac{2}{r}\bigr)
  \ = \tfrac{5}{2}-\tfrac{1}{r}-\varepsilon .
\]
Since $\dist(\MIV(P),P) \leq \dist(\MIV)$ for every  $P$, letting $r\to\infty$ and $\varepsilon\to0$
gives the claim.
\end{proof}

Together with \Cref{thm:main}, this shows that $\dist(\MIV)=5/2$.

\section{Conclusion}
We introduced Integrated Veto and showed that its equal mixture with Maximal
Lotteries has metric distortion exactly $5/2$, improving the previous upper
bound of $2.75271$.  \citet{CRWW24} ask whether Plurality Veto and its variants
could be used to improve their bound; Lemma~\ref{lem:IVbound} is one answer.

One direction to improve this bound further is to let the mixing weight depend on the profile, which our lower bound does not rule out. Another  straightforward approach is to consider mixtures of more than two rules.

\section*{Usage of AI}
Throughout the research process of finding a rule with low metric distortion, we have used GPT 5.6-Sol. After a long discussion on another approach to achieve low metric distortion it suggested the mixed integrated veto rule and its connection to maximal lotteries. Claude Opus 5.0 was used to support the write up of the paper.

\section*{Acknowledgments}
Fabian Frank is supported by the Deutsche Forschungsgemeinschaft under grant BR 2312/14-1 and would like to thank Jannik Peters and Satyanand Rammohan for feedback given on the presentation of this result.

\end{document}